\documentclass[12pt]{article}
\usepackage{amsmath}
\usepackage{amssymb,amsthm}
\usepackage{color}

\newtheorem{Thm}{Theorem}[section]
\newtheorem{theorem}[Thm]{Theorem}
\newtheorem{proposition}[Thm]{Proposition}
\newtheorem{corollary}[Thm]{Corollary}

\newtheorem{remark}{Remark}[section]

\newtheorem{definition}[Thm]{Definition}

\title{The second law of thermodynamics as a deterministic theorem for quantum spin systems}

\author{Walter F. Wreszinski\footnote{wreszins@gmail.com, 
Instituto de Fisica, Universidade de S\~ao Paulo (USP), Brasil}}        
        
\begin{document}

\maketitle

\begin{abstract}
We review our approach to the second law of thermodynamics as a theorem assering the growth of the mean (Gibbs-von Neumann) 
entropy of quantum spin systems undergoing automorphic (unitary) adiabatic
transformations. Non-automorphic interactions with the environment, although known to produce on the average
a strict reduction of the entropy of systems with finite number of degrees of freedom, are proved to conserve the mean entropy
on the average. The results depend crucially on two properties of the mean entropy, 
proved by Robinson and Ruelle for classical systems, and Lanford and Robinson for quantum lattice systems: upper semicontinuity 
and affinity. 
\end{abstract}

\section{Introduction and Summary}

In this paper, we review our approach \cite{Wre1} to the second law of thermodynamics as a theorem asserting the growth of the mean 
(Gibbs-von Neumann) entropy of a class of quantum spin systems undergoing automorphic (unitary) adiabatic transformations. The
structure of the proofs in the original reference is clarified at several points in which only a sketch is given, at the same 
time allowing for a larger class of interaction potentials (of polynomial decrease rather than only exponential decrease at infinity).
We also compare this framework with our previous approach to the second law, together with Heide Narnhofer \cite{NarWre} based on the 
(quantum) Boltzmann entropy. Our main new result is that non-automorphic interactions with the environment, although known to produce 
on the average a strict reduction of the Boltzmann entropy of systems with finite number of degrees of freedom (lemma 3 of \cite{NarWre}), 
are proved to conserve the mean entropy on the average, as a consequence of the latter's property of affinity. As a consequence, we
do not need to assume that such interactions are rare on the thermodynamic time scale, in order to assure the validity of the second
law (see the remarks on the last paragraph of \cite{NarWre}).    

As described by Wehrl in a still very readable article (\cite{Wehrl}, p. 227 et seq. ``A paradox''), the common assertion of the 
second law of thermodynamics, that the entropy of a closed system never decreases, is in striking contradiction to the following fact.
Let $S_{N,V}^{G}$ denote the Gibbs-von Neumann entropy of a system of $N$ particles in a box of volume $V$, i.e.,
\begin{equation}
S_{N,V}^{G} = -k_{B} Tr \rho_{N,V}\log(\rho_{N,V})
\label{(1.1)}
\end{equation}
where $\rho_{N,V}$ is a positive trace-class operator (density matrix) describing the system, and $Tr$ is the trace over the 
corresponding Hilbert space ${\cal H}_{N,V}$. For a system described by a time-independent Hamiltonian $H_{N,V}$, however, the
density matrix $\rho_{N,V}(t)$ at time $t$ is obtained from the density matrix at time zero $\rho_{N,V}$ from the formula
\begin{equation}
\rho_{N,V}(t) = \exp(-iH_{N,V}t) \rho_{N,V} \exp(iH_{N,V}t)
\label{(1.2)}
\end{equation}
and, since $\exp(iH_{N,V}t)$ is a unitary operator, the repeated eigenvalues of $\rho_{N,V}(t)$ are the same as the eigenvalues 
of $\rho_{N,V}$, and thus
\begin{equation}
S_{N,V}^{G}(t) = S_{N,V}^{G}(0)
\label{(1.3)}
\end{equation}
As Wehrl remarks, ``there is one way out of this dilemma: the time-evolution of a system is not described by the Schr\"{o}dinger
equation, but by some other equation. In fact, in statistical mechanics one uses, with great success, equations like the
Boltzmann equation, the master equation, and other equations.'' We refer to \cite{Sewell1} for a very comprehensive discussion of
the master equation approach, as well as references. This approach concerns open systems. In the present paper we shall be, however, concerned
with \emph{closed} systems - systems completely isolated from all external influences. Traditional thermodynamics takes this idealization
as starting point, which is also done by Lieb and Yngvason in their rigorous axiomatic approach \cite{LY}. Examples such as the evolution 
of the Universe or the adiabatic irreversible expansion of a gas, in which the isolation of the system may be achieved to any degree 
of accuracy, suggests that this idealization also has a fundamental character, from the physical point of view (see also \cite{Wre1}).

We wish to maintain the universal definition \eqref{(1.1)} of the Gibbs-von Neumann entropy, and attempt to find a dynamical proof
of the second law, continuing to assume a deterministic evolution \eqref{(1.2)}, which implies \eqref{(1.3)}. An important hint 
in this direction was provided by Penrose and Gibbs (\cite{Pen1}, p.1959, \cite{Gibbs}) in the framework of classical statistical
mechanics, whereby \eqref{(1.1)} becomes
\begin{equation}
S_{G}^{\Gamma}(\rho) = -k_{B}\int_{\Gamma} dx \rho(x) \log(\rho(x))
\label{(1.4)}
\end{equation}
(defining $x\log(x)=0$, for $x=0$, as we do throughout), where $\Gamma$ denotes phase space, and $\rho$ is a phase space density,
i.e., a nonnegative function, absolutely continuous with respect to Lebesgue measure $dx$, such that
\begin{equation}
\int_{\Gamma}dx \rho(x) = 1
\label{(1.5)}
\end{equation}
The number of particles $N$ and the box of volume $V$ are implicit in $\Gamma$; if necessary to make them explicit, we shall write
$\Gamma_{N,V}$. Define \cite{Pen1} an adiabatic process as a process in which no heat enters or leaves the system, starting at
time zero, when the Hamiltonian is $H_{t=0}$ and the phase space density $\rho_{0}$ is invariant under $H_{t=0}$, and ends at time
one, when the Hamiltonian has changed to $H_{t=1}$, supposed to yield a mixing (hence ergodic) motion thereafter. In spite of the
time-dependence of the Hamiltonian, Liouville's theorem remains valid, which permits to calculate the phase space density $\rho_{t}$
for $t \ge 0$, but the final equilibrium is assumed to be described by the coarse-grained phase-space density $\bar{\rho}$, defined
by the averaging prescription
\begin{equation}
\bar{\rho}(x) = \lim_{T \to \infty} \frac{\int_{1}^{1+T} \rho_{t} dt}{T}
\label{(1.6)}
\end{equation}
which exists for a.e. $x \in \Gamma$, by Birkhoff's ergodic theorem and which, for ergodic systems as assumed, yields an invariant
density (see, e.g., \cite{Wa}, Chap. 1, paragraph 5), just as equilibrium is modelled in Gibbs-ensemble theory. We may now state

\begin{theorem}
\label{th:1.1}
(Penrose-Gibbs theorem) Under the above assumptions,
\begin{equation}
S_{G}^{\Gamma}(\rho_{0}) \le S_{G}^{\Gamma}(\bar{\rho})
\label{(1.7)}
\end{equation}
\end{theorem}

As remarked by Penrose (\cite{Pen1}, p. 1959), mathematically, the inequality \eqref{(1.7)} illustrates the non-interchangeability
of the functional $S_{G}^{\Gamma}$, given by \eqref{(1.4)}, and the limit on the r.h.s. of \eqref{(1.6)}, since, by Liouville's 
theorem (the analogue of \eqref{(1.3)}), the l.h.s. of \eqref{(1.7)} equals $\lim_{T \to \infty} \frac{\int_{1}^{1+T}S_{G}^{\Gamma}(\rho_{t})dt}{T}$,
while the r.h.s. equals $S_{G}^{\Gamma}(\lim_{T \to \infty} \frac{\int_{1}^{1+T}\rho_{t} dt}{T})$.

In \cite{Wre1} we observed that, for finite systems, i.e., associated to a phase space $\Gamma$ of finite measure (with $N,V$ finite),
$S_{G}^{\Gamma}$, given by \eqref{(1.4)}, is a continuous functional of $\rho$ (in the natural, weak* topology, see \cite{Sewell1}, p. 57), 
and therefore, unfortunately, \eqref{(1.3)} continues to hold.

For classical systems, it is the \emph{mixing property} which determines the approach to equilibrium, in the sense
\begin{equation}
\lim_{t \to \infty} \int_{\Gamma} \rho_{t}(x) G(x) dx = <G>_{eq} \equiv \int_{\Gamma} \bar{\rho}(x) G(x) dx
\label{(1.8)}
\end{equation}
for $G$ an observable, i.e., a continuous function over phase space (see \cite{Pen1}, (1.35), p. 1949, and references given there). Equation 
\eqref{(1.8)} suggests regarding $\rho$ as a \emph{state}, i.e., a normalized linear functional over the algebra of observables, which turns
out to be the natural choice for both classical and quantum statistical mechanics, in the former case the algebra being abelian. In this
context, the notion of \emph{mean entropy}
\begin{equation}
s_{G}(\rho) = \lim_{V \to \infty} \frac{S_{G}^{\Gamma_{N,V}}(\rho)}{V}
\label{(1.9)}
\end{equation}
as a functional of the state $\rho$ is the natural quantity to be considered. In \eqref{(1.9)}, the limit $V \to \infty$, assumed to exist,
is the thermodynamic limit in the sense of van Hove, whereby $N \to \infty$, $V \to \infty$, with $\frac{N}{V}=d$, $d$ denoting the particle
density in classical statistical mechanics \cite{RR}, or the van Hove limit in the case of quantum spin systems.

The reason for considering the mean entropy $s_{G}(\rho))$ rather than the entropy $S_{G}^{\Gamma}(\rho)$ of finite systems is two-fold, and
appears in both contexts of automorphic evolutions (section 3) and non-automorphic evolutions (section 4), the later also relevant to the
measurement problem studied by Hepp in his seminal paper \cite{Hepp}.

Firstly, the limit $t \to \infty$ will not, in general, commute with the limit of infinite volume (for an explicit example of this fact
in a soluble model in cosmology describing the CMB (cosmic microwave background) radiation, see \cite{Wresz}). Indeed, the rate of approach
to equilibrium may grow indefinitely when the volume tends to infinity, and it is therefore of essential importance to consider the
\emph{densities}, such as the entropy density in \eqref{(1.9)}, or the space-average of the magnetization in the model of section 4.

Secondly, quantum systems imperatively require the use of infinite systems in the study of dynamics, since, for finite systems, the
Hamiltonian has a discrete spectrum and the observables are quasi-periodic functions of time (see section (1.5) of \cite{Pen1}).

It turns out that consideration of $s_{G}(\rho)$ rather than $S_{G}^{\Gamma_{N,V}}(\rho)$ entails two bonuses: in the first place, it is 
upper-semicontinuous rather than continuous , leading to a refined form of Theorem ~\ref{th:1.1}, in which strict
growth (of $s_{G}(\rho)$) for large times is possible (Theorem 3.1). But even more is true: the property of affinity allows
a proof of stability of the second law, in the form of Theorem 3.1, under what we believe to be a paradigmatic non-automorphic interation
with the environment (Theorem 4.1) - in sharp contrast to the situation for finite systems, in which a \emph{reduction} of the average
entropy (von Neumann or Boltzmann) occurs (Lemma 3 in \cite{NarWre}).

As observed, a state $\omega$ will be defined as a normalized, positive linear functional over the algebra ${\cal A}$ of observables of
an infinite system. Readers unfamiliar with this notion may consult the classic book \cite{Sewell1}. This seems to be the adequate
moment to make a brief interlude to explain our restriction to quantum spin systems. We believe that the forthcoming Theorem 3.1 should
have a wide domain of applicability, including classical statistical mechanics, but consideration of the momenta still poses an open problem
in the classical case (see \cite{Ru}, p. 1666). Concerning the time evolution of quantum systems, only for quantum spin systems is the
time-evolution an automorphism of the natural C* algebra of observables, a fact used in Theorem 3.1; for quantum continuous systems, 
it is not, in general, the case. In order to start our explanation, we must introduce the main concepts.

An \emph{automorphism} $\tau_{t}(A)$ of ${\cal A}$ is basic concept: it is a one-to-one mapping of ${\cal A}$ onto ${\cal A}$ which
preserves the algebraic structure. For quantum spin systems, it equals the limit in the operator norm of the time-evolutes
$\exp(iH_{\Lambda}t) A \exp(-iH_{\Lambda}t)$ as the region $\Lambda \nearrow \mathbf{Z}^{\nu}$, and one speaks of a C*-dynamical 
system $({\cal A}, \tau_{t})$. The letter $\nu$ will always denote the spatial dimension. 

Comparison with the discussion preceding Theorem ~\ref{th:1.1} suggests the following definition (which follows \cite{Pen1} closely, adapting
his discussion to a framework which includes states of infinite systems):

\begin{definition} 
\label{Definition 1.2}
Let a C*-dynamical system $({\cal A}, \tau_{t})$ be given, with $t \in [-r,\infty)$. An \emph{adiabatic transformation} consists of two successive 
steps. The \emph{first step}, called preparation of the state, starts at some $t=-r$, with $r>0$, when the state $\omega_{-r}$ is 
invariant under the automorphism $\tau_{-r}$, and ends at $t=0$. We require the cyclicity condition
\begin{equation}
\label{(1.11)}
\tau_{-r}=\tau_{0}=\tau
\end{equation}
The \emph{second step} is a dynamical evolution of the state $\omega$ in the form
\begin{equation}
t \in \mathbf{R} \to \omega_{t} \equiv \omega \circ \tau_{t}
\label{(1.12)}
\end{equation}
where the circle denotes composition, i.e., $(\omega \circ \tau_{t})(A)=\omega(\tau_{t}(A))$. We call \eqref{(1.12)} an automorphic
evolution. It is assumed to be nontrivial, i.e., the state after preparation is \emph{not} invariant under the evolution:
\begin{equation}
\label{(1.13)}
(\omega_{0} \circ \tau) \ne \omega_{0}
\end{equation}
\end{definition}

Note that the automorphism $\tau$ refers to a \emph{time-independent} interaction. Between $t=-r$ ant $t=0$ in the above definition a
``time-dependent Hamiltonian'' is supposed to act, see the above discussion of the Penrose-Gibbs theorem in the classical case. 
Indeed, according to Definition ~\ref{Definition 1.2}, the system is closed from $t=0$ to any $t>0$, but not from $t=-r$ to $t=0$, where it is subject
to \emph{external conditions}, but is still thermally isolated. One form of the second law of thermodynamics may be found in the 
seminal paper of Pusz and Woronowicz \cite{PW}: if $\omega_{-r}$ is a \emph{passive state} (defined in \cite{PW} or \cite{BRo2}, p. 101),
it is unable to perform work in a cyclic process. This statement is analogous to the Kelvin-Planck statement, rigorously formulated and
proved as Theorem 3.4 of \cite{LY}. Under certain additional conditions (\cite{PW}, or \cite{BRo2}, Theorem 5.3.22), it may be proved
that $\omega_{-r}$ is either a ground state, a thermal (KMS) state or a ceiling (infinite temperature) state. We shall refer to any of
the latter states as an \emph{equilibrium state}, and chose our $\omega_{-r}$ in Definition ~\ref{Definition 1.2} as an equilibrium state.
As remarked in \cite{BRo2}, p. 211, (5.4.4), passivity ``reflects a property of stability of equilibrium states which is basically kinematic''.
It does not include the second step in Definition ~\ref{Definition 1.2}, which is basically dynamic and may be called a relaxation. For finite
systems, this process conserves the entropy, but we shall see that, for systems with infinite number of degrees of freedom, it may lead
to growth of the mean entropy. In the explicit examples we are able to provide, the mean entropy is conserved and equals zero in the
first step, while the growth appears in second step; it is a characteristic of states of infinite systems that pure states may tend
to mixed states for large times, a fact observed by Hepp \cite{Hepp}.

Equation \eqref{(1.13)} is a condition of \emph{dynamical instability}. We shall correspondingly refer to $\omega_{0}$ as an \emph{unstable state}.

\begin{remark}
\label{Remark 1.1}
The ``time-arrow'' problem.

In Definition ~\ref{Definition 1.2} we see that primitive causality requires a time interval $r \ne 0$ preceding the time of appearance of
the state $\omega_{0}=\omega$ during which the state is prepared. This is obviously not time-reversal invariant, and this mere fact implies
the existence of an ``arrow of time'', see also \cite{Wre1} and \cite{Pe}. We shall implicitly assume it in the present paper. See also
section 4.4.
\end{remark}

\section{Quantum spin systems}
 
\subsection{Generalities}

A prototype, which we shall use, is the generalized Heisenberg Hamiltonian (generalized Ising model (gIm) if $J_{1}=0$):
\begin{equation}
\label{(2.1)}
H_{\Lambda} = -2\sum_{x,y \in \Lambda}[J_{1}(x-y)(S_{x}^{1}S_{y}^{1}+S_{x}^{2}S_{y}^{2})+J_{2}(x-y)S_{x}^{3}S_{y}^{3}]
\end{equation}
where
\begin{equation}
\label{(2.2)}
\sum_{x \in \mathbf{Z}^{\nu}}|J_{i}(x)|< \infty \mbox{ and } J_{i}(0)=0 \mbox{ for } i=1,2
\end{equation}
Above, $\vec{S}_{x} \equiv (S_{x}^{1},S_{x}^{2},S_{x}^{3})$, where $S_{x}^{i}=1/2 \sigma_{x}^{i}, i=1,2,3$ and $\sigma_{x}$ are the
Pauli matrices at the site $x$. Above, $H_{\Lambda}$ acts on the Hilbert space ${\cal H}_{\Lambda}=\otimes_{x \in \Lambda}\mathbf{C}_{x}^{2}$,
and $\vec{S}_{x}$ is short for $\mathbf{1} \otimes \cdots \otimes \vec{S}_{x} \otimes \cdots \otimes \mathbf{1}$.
We define
${\cal A}(\Lambda) = B({\cal H}_{\Lambda})$
These local algebras satisfy:
1.) Causality: $[{\cal A}(B),{\cal A}(C)]=0$ if $B \cap C = \phi$ and 2.) Isotony: $B \subset C \Rightarrow {\cal A}(B) \subset {\cal A}(C)$.
${\cal A}_{L} = \cup_{B} {\cal A}(B)$ is termed the \emph{local} algebra; its closure with respect to the norm, the \emph{quasilocal} algebra
(observables which are, to arbitrary accuracy, approximated by observables attached to a \emph{finite} region). The norm is
$A \in B({\cal H}_{\Lambda}) \to ||A|| = sup_{||\Psi|| \le 1} ||A \Psi||$, $\Psi \in {\cal H}_{\Lambda}$.
An automorphism $\tau_{t}(A) \mbox{ equal to the norm limit of } \lim_{\Lambda \nearrow \infty} \exp(iH_{\Lambda}t) A \exp(-iH_{\Lambda}t)$ 
for $A \in {\cal A}(\Lambda)$. The limit $\lim_{\Lambda \nearrow \infty}$ will denote the van Hove limit \cite{BRo2}, p.287).
 
For a finite quantum spin system the Gibbs-von Neumann entropy is ($k_{B}=1$)
\begin{equation}
\label{(2.3)}
S_{\Lambda} = -Tr (\rho_{\Lambda} \log \rho_{\Lambda})
\end{equation}
We may view $\rho_{\Lambda}$ as a \emph{state} $\omega_{\Lambda}$ on ${\cal A}(\Lambda)$ - a positive, normed linear functional on ${\cal A}(\Lambda)$:
$\omega_{\Lambda}(A) = Tr_{{\cal H}_{\Lambda}} (\rho_{\Lambda} A) \mbox{ for } A \in {\cal A}(\Lambda)$
(positive means $ \omega_{\Lambda}(A^{\dag}A) \ge 0$, normed $\omega_{\Lambda}(\mathbf{1})=1$.)

\subsection{Finite versus infinite systems}

This notion of state generalizes to systems with
infinite number of degrees of freedom $\omega(A)= \lim_{\Lambda \nearrow \infty} \omega_{\Lambda}(A)$, at first for $A \in {\cal A}_{L}$
and then to ${\cal A}$.
The state $\omega_{\Lambda}^{t}(A) = \omega_{\Lambda}(\exp(iH_{\Lambda}t)A\exp(-iH_{\Lambda}t))$ does \emph{not} have a limit
as $\Lambda \nearrow \infty$ because the spectrum is discrete and the state an almost periodic function of $t$.

As we shall see later, for states of infinite systems, however, the situation is \emph{entirely different}!

The function $S_{\Lambda}^{t} \equiv S_{\Lambda}(\omega_{\Lambda}^{t}) = -Tr(\rho_{\Lambda}^{t} \log \rho_{\Lambda}^{t}) = S_{\Lambda}^{0}$.
It is a \emph{continuous} functional of the state $\omega_{\Lambda}^{t}$.
For a large system the \emph{mean entropy} is the natural quantity from the physical standpoint:
\begin{equation}
\label{(2.4)}
s(\omega) \equiv \lim_{\Lambda \nearrow \infty} (\frac{S_{\Lambda}}{|\Lambda|})(\omega)
\end{equation}
The mean entropy has the following two properties \cite{LanRo}:
\begin{itemize}
\item [$a.)$] $0 \le s(\omega) \le \log D$ where $D=2S+1$;
\item [$b.)$] $s$ is \emph {upper semicontinuous}, that is $lim sup_{n \to \infty} s(\omega_{n}) \le s(\omega)$.
\end{itemize}
For b.) see Lemma I.1 of Appendix 2 of \cite{BB}: $\omega_{n}$ is a sequence of states such that $\omega_{n} \to \omega$ 
in the weak*- topology, i.e., $\omega_{n} (A) \to \omega(A) \forall A \in {\cal A}$. One simple example of b.) is a 
characteristic function of a closed set. For the statistical thermodynamical significance of this property, se \cite{Sewell1}, p.55.

Let $\Gamma = \mathbf{Z}^{\nu}$ and $P_{0}(\Gamma)$ denote the collection of all finite parts of $\Gamma$. We consider in this
paper quantum spin systems described by a Hamiltonian for any finite region $\Lambda \subset \mathbf{Z}^{\nu}$
\begin{equation}
\label{(2.5)}
H_{\Phi}(\Lambda) = \sum_{X \subset \Lambda} \Phi(X)
\end{equation}
where $\Phi$ is a translation-invariant interaction function from $P_{0}(\Gamma)$ into the Hermitian elements of the quasi-local
algebra ${\cal A}$, such that $\Phi(X) \in {\cal A}(X)$ and
\begin{equation}
\label{(2.6)}
\Phi(X+x) = \tau_{x}(\Phi(X)) \forall X \subset \mathbf{Z}^{\nu}, \forall x \in \mathbf{Z}^{\nu}
\end{equation}
Although general quantum lattice systems could be treated, including Fermion lattice systems, we shall for simplicity restrict
ourselves to quantum spin systems with two-body interactions. In this case, stability requires that
\begin{equation}
\label{(2.7)}
\sum_{0 \ne x \in \mathbf{Z}^{\nu}} ||\Phi(\{0,x\})|| < \infty
\end{equation}

If, in \eqref{(2.1)}, we consider the Ising limit, i.e., set $J_{1} \equiv 0$, and define $J_{2}=J$, we obtain the generalized
Ising model (gIm) studied by Emch \cite{Em1} and Radin \cite{Ra} (we changed the overall sign to agree with Radin's notation)
\begin{equation}
\label{(2.8)}
H_{\Lambda} = \frac{1}{2} \sum_{x,y \in \Lambda} J(|x-y|) \sigma_{x}^{3}\sigma_{y}^{3}
\end{equation}
with
\begin{equation}
\label{(2.9)}
J(0)=0 \mbox{ and } \sum_{x \in \mathbf{Z}^{\nu}} |J(|x|)| < \infty
\end{equation} 
Let, now, $\nu=1$, and define two subclasses of models of the gIm:
\begin{equation}
\label{(2.10)}
J(|x|) = \xi^{-|x|} \mbox{ with } \xi > 1 \mbox{ and } x \ne 0 (\mbox{ \emph{exponential model} } E_{\xi})
\end{equation}
and
\begin{equation}
\label{(2.11)}
J(|x|) = \frac{1}{|x|^{\alpha}} \mbox{ with } \alpha > 1 \mbox{ and } x \ne 0 (\mbox{ \emph{Dyson model} } D_{\alpha})
\end{equation}
The above conditions on $\xi$ and $\alpha$ are required for stability \eqref{(2.9)}. Under certain conditions, the Dyson
model displays a ferromagnetic phase transition \cite{Dy}. 
With $P_{0}(\mathbf{Z}^{\nu})$ denoting the set of all finite subsets of $\mathbf{Z}^{\nu}$ as before, for each triple 
$A=(A_{1},A_{2},A_{3})$, where the $A_{i} \in P_{0}(\mathbf{Z}^{\nu}, i=1,2,3$ are pairwise disjoint, define
\begin{equation}
\label{(2.12)}
\sigma^{A} \equiv \prod_{x_{1}\in A_{1}} \sigma_{x_{1}}^{1} \prod_{x_{2}\in A_{2}} \sigma_{x_{2}}^{2} \prod_{x_{3}\in A_{3}} \sigma_{x_{3}}^{3} 
\end{equation}
Let $\omega$ be any state satisfying
\begin{equation}
\label{(2.13)}
\omega(\sigma^{A}) = 0 \forall A \mbox{ such that } A_{3} \ne \phi
\end{equation}

We have the following:

\begin{theorem}
\label{th:2.1}
For the gIm with interaction either of the exponential model $E_{\xi}$ with $\xi$ a transcendental number, or the Dyson
model $D_{\alpha}$, and any initial state (at $t=0$) $\omega$ satisfying \eqref{(2.13)},
\begin{equation}
\label{(2.14)}
\lim_{t \to \infty} (\omega \circ \tau_{t}) = \omega_{eq} \equiv \otimes_{x \in \mathbf{Z}} Tr_{x} 
 \end{equation}
where the limit on the l.h.s. of \eqref{(2.14)} is taken in the weak* topology.
\end{theorem}

For the proof, see (\cite{Ra}, Corollary, p.2953).

\subsection{The Lieb-Robinson bound and a limit theorem of Nachtergaele, Sims and Young on local approximations to time-translation automorphisms}

Quantum spin systems are similar to quantum fields because of the now famous Lieb-Robinson bound (\cite{LR}, see also \cite{BRo2}, p.254).
Following the last reference, let us assume, instead of \eqref{(2.7)},
\begin{equation}
\label{(2.15)}
||\Phi||_{\lambda} \equiv \sum_{x \in \mathbf{Z}^{\nu}} ||\Phi(\{0,x\})|| \exp(\lambda |x|) < \infty
\end{equation}
for some $\lambda>0$. Then, for $A,B \in {\cal A}_{0}$,
\begin{equation}
\label{(2.16)}
||[(\tau_{x}\tau_{t})(A),B]|| \le 2||A||||B|| \exp[-|t|(\lambda |x|/|t|-2 ||\Phi||_{\lambda})]
\end{equation}
In \eqref{(2.16)},$(\tau_{x}\tau_{t})$ may be replaced by $(\tau_{t}\tau_{x})$. The commutator $[(\tau_{x}\tau_{t})(A),B]$ with
$B \in {\cal A}_{0}$ provides a measure of the dependence of the observation $B$ at the point $x$ at time $t$ at the origin at time $t=0$,
showing that this effect decreases exponentially with time outside the cone
\begin{equation}
\label{(2.17)}
|x| < |t| (\frac{2||\Phi||_{\lambda}}{\lambda}) 
\end{equation}
Equation \eqref{(2.17)} means that physical disturbances propagate with ``group velocity'' bounded by
\begin{equation}
\label{(2.18)}
v_{\Phi} \equiv \inf_{\lambda} (\frac{2||\Phi||_{\lambda}}{\lambda})
\end{equation}

Such bounds have been considerably developed in the last thirty years, and the recent reference by Nachtergaele, Sims and Young
\cite{NSY} contains a comprehensive summary of some of these results, as well as an extensive list of references to earlier work on 
the subject.

Even more significant for us in section 3 are new results in \cite{NSY} on local approximations of time-translation automorphisms
of local observables, which hold even under conditions weaker than \eqref{(2.15)}, for which the interpretation of the Lieb-Robinson
bounds in terms of a bounded ``group velocity'' of physical disturbances \eqref{(2.16)}, \eqref{(2.17)} no longer applies.
We now assume (only) \eqref{(2.7)} and follow \cite{NSY} to introduce the notion of an $F$ function on the metric space
$(\Gamma, d)$, with $d$ a distance function, in order to quantify the locality properties of an interaction. For definitiveness,
$\Gamma = \mathbf{Z}^{\nu}$, ${\cal A}_{\Gamma}^{loc}$ denotes a local algebra, i.e., ${\cal A}_{X}$ with $X \in P_{0}(\Gamma)$, with,
as before, $P_{0}(\Gamma)$ denoting the collection of all finite parts of $\mathbf{Z}^{\nu}$, and $d(x,y) \equiv \sum_{i=1}^{\nu}|x_{i}-y_{i}|$,
for $x=(x)_{i}, y=(y_{i}), i=1, \cdots \nu$. $F$ is assumed to be a non-increasing function $F:[0,\infty) \to (0,\infty)$ satisfying
the properties:
(i) $F$ is uniformly integrable over $\Gamma$, i.e.,
\begin{equation}
\label{(3.5.5)}
||F|| \equiv \sup_{x \in \Gamma} \sum_{y \in \Gamma} F(d(x,y)) < \infty
\end{equation}
(ii) $F$ satisfies the convolution condition
\begin{equation}
\label{(3.5.6)}
C_{F} \equiv \sup_{x,y \in \Gamma} \sum_{z \in \Gamma} \frac{F(d(x,z)F(d(z,y))}{F(d(x,y))} < \infty
\end{equation}
An equivalent formulation of \eqref{(3.5.6)} is that there is a constant $C_{F}< \infty$ such that
\begin{equation}
\label{(3.5.7)}
\sum_{z \in \Gamma} F(d(x,z)) F(d(z,y)) \le C_{F} F(d(x,y)) \forall x,y \in \Gamma
\end{equation}
The $F$ functions describe the decay of a given interaction. Let $F$ be an $F$- function on $(\Gamma,d)$ and 
$\Phi : P_{0}(\Gamma) \to {\cal A}_{\Gamma}^{loc}$ be an interaction. The \emph{$F$- norm of $\Phi$} is defined by
\begin{equation}
\label{(3.5.8)}
||\Phi||_{F} = \sup_{x,y \in \Gamma} \frac{1}{d(x,y)}\sum_{Z \in P_{0}, x,y \in Z} ||\Phi(Z)||
\end{equation}
from which
\begin{equation}
\label{(3.5.9)}
\sum_{Z \in \P_{0}, x,y \in Z} ||\Phi(Z)|| \le ||\Phi||_{F} F(d(x,y))
\end{equation}

An important question will arise in section 3 in connection with the following issue: automorphisms of strictly local observables
delocalize along the quasi-local algebra ${\cal A}$, as seen explicitly in (\cite{Ra}, p.2951) in the limit
$ \Lambda \nearrow \infty$, but there are (again by the same formulae) local approximations to them. The following theorem will
be of crucial importance for us in section 3:

\begin{theorem}
\label{th:3.5.1}
Let $\Lambda_{n} \subset \Lambda_{n+1}, \forall n \ge 1$ denote an exhaustive sequence of increasing and absorbing finite regions
$\{\Lambda\}_{n}$, denoted $\Lambda_{n} \nearrow \Gamma$, $A$ be an element of the strictly local algebra ${\cal A}_{\Gamma}^{loc}$, and
\begin{equation}
\label{(3.5.10)}
\tau_{t}^{\Lambda} = \exp(iH_{\Lambda}t) A \exp(-iH_{\Lambda}t)
\end{equation}
with $H_{\Lambda}$ defined by \eqref{(2.5)}, under the condition
\begin{equation}
\label{(3.5.11)}
||\Phi||_{F} < \infty 
\end{equation}
for a given $F$- fuction $F$. Then,
\begin{equation}
\label{(3.5.12)}
||\tau_{t}^{\Lambda_{n}}(A)-\tau_{t}^{\Lambda_{m}}(A)|| \le I_{t}(\Phi)\sum_{x \in X}\sum_{y \in \Lambda_{n}-\Lambda_{m}}F(d(x,y))
\end{equation}
where
\begin{equation}
\label{(3.5.13)}
I_{t}(\Phi) \equiv \frac{2||A||}{C_{F}} (2C_{F}|t|||\Phi||_{F})\exp(2C_{F}|t|||\Phi||_{F})
\end{equation}
Above, we may assume that $A \in {\cal A}_{X}$ for some $X \in P_{0}(\mathbf{Z}^{\nu})$, and, since the sequence is exhaustive,
$\exists N \ge 1$ for which $X \subset \Lambda_{n} \forall n \ge N$. The integers $n,m$ in \eqref{(3.5.15)} satisfy $N \le m \le n$.
Further
\begin{equation}
\label{(3.5.14)}
\tau_{t}(A) = \lim_{\Lambda \nearrow \infty} \tau_{t}^{\Lambda}(A) \forall A \in {\cal A}_{\Gamma}^{loc}
\end{equation}
exists in norm, and the convergence is uniform for $t$ in compact subsets of $\mathbf{R}$.
\end{theorem}

\begin{proof}

See Theorem 3.4 and (3.78) of Theorem 3.5 of \cite{NSY}. The r.h.s. of (38) converges to zero, as $n,m \to \infty$,
by \eqref{(3.5.5)}, assumption \eqref{(3.5.11)} and definition \eqref{(3.5.8)}.

\end{proof}

We have the following basic corollary:
\begin{corollary}
\label{cor:3.5.1}
\emph{Theorem of Nachtergaele, Sims and Young on local approximations to automorphisms of local observables}

Under the assumptions of Theorem ~\ref{th:3.5.1},
\begin{equation}
\label{(3.5.15)}
||\tau_{t}^{\Lambda_{n}}(A)-\tau_{t}^{\Lambda_{m}}(A)|| \le I_{t}(\Phi) \sum_{x \in X} \sum_{y \in \Gamma - \Lambda_{m}}F(d(x,y))
\end{equation}

\end{corollary}

We shall assume that our two-body translation-invariant interactions satisfy the condition $P$ below:

$P$  The interactions $\Phi$ have polynomial decay at infinity, i.e., for some $0<\beta<\infty$,
\begin{equation}
\label{(3.5.16)}
|\Phi(\{0,x\})| \le \beta (1+|x|)^{-\alpha}
\end{equation}
where $|x|=d(x,0)$ and
\begin{equation}
\label{(3.5.17)}
\alpha > \nu
\end{equation}
Condition \eqref{(3.5.17)} is due to the stability requirement \eqref{(2.7)}.

In our applications in section 3 we shall use Corollary 2.3. The r.h.s. of \eqref{(3.5.15)} converges to zero, as $m \to \infty$,
under assumption \eqref{(3.5.11)}, but the set $X$ in \eqref{(3.5.15)} must be also finally made to grow indefinitely, due
to the definition of the mean entropy. This will require condition $P$ for suitable $\alpha$, as we shall see in connection with
the forthcoming inequality of Fannes.

\subsection{A theorem of Fannes and a main auxiliary theorem}

The entropy satisfies the subadditivity property \cite{LanRo}  
\begin{equation}
\label{(3.6.8)}
S_{\Lambda_{1}\cup \Lambda_{2}}(\omega) \le S_{\Lambda_{1}}(\omega) + S_{\Lambda_{2}}(\omega) 
\end{equation}
for $\Lambda_{1} \cap \Lambda_{2} = \phi$. 

Of great importance in section 3 is the following theorem, due to Fannes \cite{Fa}, let $\omega_{1}$ and $\omega_{2}$ be two states
on ${\cal A}$, i.e., $\omega_{1}, \omega_{2} \in E_{{\cal A}}$, and, for $\Lambda_{0} \subset \mathbf{Z}^{\nu}$, let $\omega_{1,\Lambda_{0}}$
and $\omega_{2,\Lambda_{0}}$ denote the corresponding density matrices. denote by $\lambda_{k}^{1}, \lambda_{k}^{2}, k=1, \cdots, D^{|\Lambda_{0}|}$
the repeated eigenvalues of $\omega_{1,\Lambda_{0}}$ and  $\omega_{2,\Lambda_{0}}$, respectively, in ascending order. Define
\begin{equation}
\label{(3.6.16)}
\epsilon_{k} = |\lambda_{k}^{1}-\lambda_{k}^{2}|
\end{equation}
and
\begin{equation}
\label{(3.6.17)}
a = \sum_{k} \epsilon_{k}
\end{equation}

\begin{theorem}
\label{th:3.6.3}
\begin{equation}
\label{(3.6.18)}
i.) a \le \sup_{||A|| \le 1, A \in {\cal A}_{\Lambda_{0}}}|Tr[(\omega_{1,\Lambda_{0}}-\omega_{2,\Lambda_{0}})A|
\end{equation}
\begin{equation}
\label{(3.6.19)}
ii.) |S_{\Lambda_{0}}(\omega_{1})-S_{\Lambda_{0}}(\omega_{2})|\le |\Lambda_{0}| a \log D + e
\end{equation}
\begin{equation}
\label{(3.6.20)}
a \le ||\omega_{1}-\omega_{2}||
\end{equation}
By \eqref{(3.6.19)} and \eqref{(3.6.20)},
\begin{equation}
\label{)3.6.21)}
|s(\omega_{1})-s(\omega_{2})| \le ||\omega_{1}-\omega_{2}|| \log D
\end{equation}

\begin{proof}

See \cite{Fa}. Equation \eqref{(3.6.18)} is the last inequality of the second line of his (1), and \eqref{(3.6.19)} is the last
inequality before the final statement in \cite{Fa}. Due to \eqref{(3.6.20)}, which is also shown in \cite{Fa}, it follows that
$0 \le a \le 2$, and the constant in \eqref{(3.6.19)} is thus bounded by $\sup_{0 \le a \le 2} (2a-a \log a)= e$.

\end{proof}
\end{theorem}

Theorem ~\ref{th:3.6.3} has some important consequences:

\begin{proposition}
\label{prop:3.6.4}
Let $\Lambda_{0}$ be a finite fixed subset of $\mathbf{Z}^{\nu}$, $\Lambda \supset \Lambda_{0}$, and
\begin{equation}
\label{(3.6.22)}
\omega^{\Lambda,t} \equiv \exp(-iH_{\Lambda}t) \omega_{\Lambda} \exp(iH_{\Lambda}t)
\end{equation}
Then, the following inequality holds:
\begin{equation}
\label{(3.6.23.1)}
|S((\omega_{t})_{\Lambda})-S(\omega^{\Lambda,t})|\le |\Lambda_{0}|\log D \sup_{A \in {\cal A}_{\Lambda_{0}},||A||=1}||\tau_{t,\Lambda,A}||+c_{\Lambda}
\end{equation}
where
\begin{equation}
\label{(3.6.23.2)}
c_{\Lambda} \equiv e + 2|\Lambda - \Lambda_{0}| \log D
\end{equation}
and $A - B$ denotes the complement of $B$ in $A$. Above,
\begin{equation}
\label{(3.6.24.1)}
\omega_{t} \equiv \omega \circ \tau_{t}
\end{equation}
and
\begin{equation}
\label{(3.6.24.2)}
\tau_{t,\Lambda,A} \equiv \tau_{t}(A)-\exp(iH_{\Lambda}t) A \exp(-iH_{\Lambda}t)
\end{equation}

\begin{proof}

By the subadditivity property of the entropy \eqref{(3.6.8)} and property 1) of the mean entropy
\begin{equation}
\label{(3.6.25)}
|S((\omega_{t})_{\Lambda})-S(\omega^{\Lambda,t})| \le |S((\omega_{t})_{\Lambda_{0}})-S(\omega^{\Lambda,t}_{\Lambda_{0}})| +2|\Lambda -\Lambda_{0}|\log D
\end{equation}
where $\omega^{\Lambda,t}_{\Lambda_{0}}$ denotes the restriction of $\omega^{\Lambda,t}$ to ${\cal A}_{\Lambda_{0}}$.
By \eqref{(3.6.19)},
\begin{equation}
\label{(3.6.26.1)}
|S((\omega_{t})_{\Lambda_{0}})-S(\omega^{\Lambda,t}_{\Lambda_{0}})| \le |\Lambda_{0}|d_{\Lambda}\log(D) + e
\end{equation}
with
\begin{equation}
\label{(3.6.26.2)}
d_{\Lambda} \equiv \sup_{A \in {\cal A}_{\Lambda_{0}}, ||A|| \le 1}|Tr_{{\cal H}_{\Lambda_{0}}}[((\omega_{t})_{\Lambda_{0}}-(\omega^{\Lambda,t}_{\Lambda_{0}})A]|
\end{equation}
But
\begin{eqnarray*}
\sup_{A \in {\cal A}_{\Lambda_{0}}, ||A|| \le 1}|Tr_{{\cal H}_{\Lambda_{0}}}[((\omega_{t})_{\Lambda_{0}}-\omega^{\Lambda,t}_{\Lambda_{0}})A]| \le \\ 
\le \sup_{A \in {\cal A}_{\Lambda_{0}}, ||A|| \le 1}|Tr_{{\cal H}_{\Lambda_{0}}}[\omega_{\Lambda}(\lim_{n \to \infty} \sigma_{n,\Lambda,A}-\\
-\exp(itH_{\Lambda}) A \exp(-itH_{\Lambda})] \\
\mbox{ where } \sigma_{n,\Lambda,A} \equiv Tr_{\Lambda_{n}-\Lambda}[\omega_{\Lambda_{n}-\Lambda} (\exp(itH_{\Lambda}) A \exp(-itH_{\Lambda})]
\end{eqnarray*}
The above inequality yields, together with \eqref{(3.6.25)}, \eqref{(3.6.26.1)}, and \eqref{(3.6.26.2)}, the equations \eqref{(3.6.23.1)} and
\eqref{(3.6.23.2)}.

\end{proof}
\end{proposition}

We now use Proposition ~\ref{prop:3.6.4}, together with \eqref{(3.5.15)}, to arrive at the main auxiliary theorem:

\begin{theorem}
\label{th:3.6.5}
Let $\omega$ be a translation-invariant state on ${\cal A}$, with a dynamics $\tau_{t}; t \in \mathbf{R}$ generated by a 
family of Hamiltonians \eqref{(2.5)}, satisfying \eqref{(2.6)} and \eqref{(2.7)}. If the interactions $\Phi$ have polynomial
decay at infinity \eqref{(3.5.16)}, with
\begin{equation}
\label{(3.6.29)} 
\alpha > 2 \nu
\end{equation}
it follows that the mean entropy $s(\omega_{t})$ of $\omega_{t}$, given by \eqref{(3.6.24.1)} satisfies
\begin{equation}
\label{(3.6.30)}
s(\omega_{t}) = \lim_{m \to \infty} \frac{1}{|\Lambda_{0,m}|}S(\exp(-itH_{\Lambda_{0,m}}) \omega_{\Lambda_{0,m}} \exp(itH_{\Lambda_{0,m}}))
\end{equation}
where $\{\Lambda_{0,m}\}$ is a sequence of boxes tending to $\mathbf{Z}^{\nu}$ as $m \to \infty$.

\begin{proof}

By definition,
\begin{equation}
\label{(3.6.31)}
s(\omega_{t}) = \lim_{m \to \infty} \frac{1}{|\Lambda_{0,m}|} S((\omega_{t})_{\Lambda_{0,m}})
\end{equation}
We now pick $\Lambda_{0} = \Lambda_{0,m} \subset \Lambda=\Lambda_{m}$ in \eqref{(3.6.23.1)} and \eqref{(3.6.23.2)}, where, for each $m \ge 1$,
\begin{equation}
\label{(3.6.32)}
\lim_{m \to \infty} \frac{|\Lambda_{0,m}|}{|\Lambda_{m}|} = 1 
\end{equation}
and
\begin{equation}
\label{(3.6.33)}
\lim_{m \to \infty} \frac{|\Lambda_{m} - \Lambda_{0,m}|}{|\Lambda_{0,m}|} = 0
\end{equation}
By \eqref{(3.6.23.1)} and \eqref{(3.6.23.2)} of Proposition ~\ref{prop:3.6.4}, we obtain the inequality
\begin{equation}
\label{(3.6.34.1)}
|S((\omega_{t})_{\Lambda_{m}})-S(\omega_{\Lambda_{m},t})| \le |\Lambda_{0,m}| (\log D)  f_{m}
\end{equation}
where
\begin{equation}
\label{(3.6.34.2)}
f_{m} \equiv \sup_{A \in {\cal A}_{\Lambda_{0,m}},||A||\le 1}||\tau_{t,\Lambda_{m},A}||+ 2|\Lambda_{m}-\Lambda_{0,m}| \log D
\end{equation}
where, by definition \eqref{(3.6.24.1)}, \eqref{(3.6.24.2)},
\begin{equation}
\label{(3.6.35)}
||\tau_{t,\Lambda_{m},A}||= ||\tau_{t}(A) - \tau_{t}^{\Lambda_{m}}(A)||
\end{equation} 
We now use \eqref{(3.5.16)} and \eqref{(3.6.29)}. A suitable $F$- function in this case is (\cite{NSY}, Appendix, (A.9))
\begin{equation}
\label{(3.6.36)}
F(|x|) = (1+|x|)^{-2\nu - \epsilon}
\end{equation}
for some $\epsilon > 0$, with $|x|=d(x,0)$. By \eqref{(3.5.5)}, \eqref{(3.6.36)} and \eqref{(3.6.32)}, it readily follows
from \eqref{(3.5.15)} that
\begin{equation}
\label{(3.6.37)}
\lim_{m \to \infty} \sup_{A \in {\cal A}_{\Lambda_{0,m}}, ||A|| \le 1}||\tau_{t}(A) - \tau_{t}^{\Lambda_{m}}(A)|| = 0
\end{equation}
Putting, now, together \eqref{(3.6.23.1)}, \eqref{(3.6.23.2)}, \eqref{(3.6.33)} and \eqref{(3.6.37)}, we obtain \eqref{(3.6.30)}.

\end{proof}
\end{theorem}

It is to be noted that \eqref{(3.6.29)} is stronger than the stability condition \eqref{(3.5.17)}. The existence of local approximations
to automorphisms of fixed local observables holds, by \eqref{(3.5.15)}, for any stable interaction, but \eqref{(3.6.30)}, which will
be crucially needed in section 3, requires a sufficiently fast decay of the interactions.

\section{The second law of thermodynamics as a theorem asserting the growth of the mean entropy under adiabatic transformations}

\subsection{A refined form of the Penrose-Gibbs theorem 1.1} 

The Penrose-Gibbs Theorem ~\ref{th:1.1} may now be stated in a refined form, which incorporates the assumption of an infinite number
of degrees of freedom:

\begin{theorem}
\label{th:4.1}
\emph{A statistical-thermodynamical version of the second law of thermodynamics}
Let a quantum spin system, described by a finite-region Hamiltonian \eqref{(2.5)}, and the state $\omega$ fulfill the assumptions
of Theorem ~\ref{th:3.6.5}. Assume further that $\omega_{t} \equiv \omega \circ \tau_{t}$ satisfies the condition
\begin{equation}
\label{(4.1.1)}
\lim_{t \to \infty} \omega_{t} = \bar{\omega}
\end{equation}
where the limit is taken in the weak*-topology, that the mean entropy $s$ exists and is defined by \eqref{(2.4)}, and denote
the initial state by $\omega_{0}= \omega$.  Then
\begin{equation}
\label{(4.1.2)}
s(\omega) \le s(\bar{\omega})
\end{equation}
In words: the mean entropy of the initial state may increase, but cannot decrease, towards that of the ``coarse-grained'' state
$\bar{\omega}$ under an adiabatic transformation (that is, with $\omega_{t}$ in \eqref{(4.1.1)} defined as in Definition ~\ref{Definition 1.2}).

\begin{proof}

We consider the sequence $\{\omega_{n}\}_{n=0,1,2, \cdots}$, which a fortiori satisfies 
\begin{equation}
\label{(4.1.3)}
\omega_{n} \to \bar{\omega} \mbox{ as } n \to \infty
\end{equation}
in the weak*-topology. By property b.) of the mean entropy, $s$ is upper semicontinuous in the weak*-topology and therefore
\begin{equation}
\label{(4.1.4)}
\lim sup_{n \to \infty} s(\omega_{n}) \le s(\lim_{n \to \infty} \omega_{n}) = s(\bar{\omega}) 
\end{equation}
where the limit on the r.h.s. of \eqref{(4.1.4)} is taken in the weak*-topology. By \eqref{(3.6.30)} and the invariance of the 
trace $Tr_{{\cal H}_{\Lambda_{0,m}}}$, for each $m \ge 1$, under unitary transformations on the space ${\cal H}_{\Lambda_{0,m}}$,
it follows that $s(\omega_{n}) = s(\omega_{0}) = s(\omega)$, and \eqref{(4.1.2)} results from \eqref{(4.1.4)}.

\end{proof}
\end{theorem}

\begin{remark}
\label{Remark 4.1}
Theorem ~\ref{th:4.1} answers in the affirmative the questions raised by Ruelle in (\cite{Ru}, (iv), p.1666).
\end{remark}

\subsection{Applications to the gIm}

Th gIm corresponds to the strong interaction limit in the language of the lattice gas (see, e.g., \cite{BRo2}, p. 425):
we expect that \emph{interactions} are crucial for the approach (or return) to equilibrium! Therefore, in spite of being
special, the present case does have certain physically sound aspects. Another example of this feature is the fact that
the special case of the gIm considered in \cite{Em1} describes well certain experiments of free induction decays in
solids \cite{LoNo}. This dynamics also received attention more than forty years later \cite{Cho}, and even quite recently
in \cite{Maze}. 

We consider the two subclasses of models of the gIm for $\nu = 1$, the exponential model $E_{\xi}$ \eqref{(2.10)} (which a fortiori 
satisfies \eqref{(3.5.16)}, \eqref{(3.6.29)}), and the Dyson model $D_{\alpha}$ \eqref{(2.11)}, with $\alpha$ satisfying \eqref{(3.6.29)}.
Then (26) holds, which is \eqref{(4.1.1)}, with
\begin{equation}
\label{(4.2.1)}
\bar{\omega} \equiv \otimes_{x \in \mathbf{Z}} Tr_{x}
\end{equation}
as long as $\omega = \omega_{0}$ satisfies (25). As a prototypical example, take
\begin{equation}
\label{(4.2.2)}
\omega_{0} = \otimes_{x \in \mathbf{Z}} P_{x}^{+,1} 
\end{equation}
where
\begin{equation}
\label{(4.2.3)}
\sigma_{x}^{1} P_{x}^{+,1} = P_{x}^{+,1}
\end{equation}
We have
\begin{corollary}
\label{cor:4.2.1}
For the exponential model $E_{\xi}$ \eqref{(2.10)} and the Dyson model $D_{\alpha}$ \eqref{(2.11)} with $\alpha > 2$, and
$\omega_{0}$ given by \eqref{(4.2.2)}, Theorem ~\ref{th:4.1} is fulfilled in a nontrivial way, with $s(\omega_{0})=0$ and
$s(\bar{\omega})= \log 2$ (with $S=1/2$).

\begin{proof}

The finite-volume version satisfies
$$
\frac{S_{\Lambda}(\omega_{0})}{|\Lambda|} = (- \lambda \log \lambda)_{\lambda=0} = 0
$$
for all $\Lambda \subset \mathbf{Z}$, hence $s(\omega_{0}) = 0$. The finite volume version
$$
\frac{S_{\Lambda}(\bar{\omega})}{|\Lambda|} = (-\sum_{i=1}^{2}\lambda_{i} \log \lambda_{i})_{\lambda_{1}=\lambda_{2}=1/2}= \log 2
$$
for all $\Lambda \subset \mathbf{Z}$, hence  $s(\bar{\omega}) = \log 2$.

\end{proof}
\end{corollary}

\begin{remark}
\label{Remark 4.2}
Note that the two values $s(\omega_{0})=0$ and $s(\bar{\omega}) = \log 2$ correspond to the two extreme values in property 1) of the
mean entropy. Thus $\bar{\omega}$ corresponds to a maximum of the mean entropy, which is identified with the ``infinite temperature state''.
The initial state $\omega_{0}$ may be seen as a nonlocal perturbation of the equlibrium state $\bar{\omega}$, if \eqref{(4.1.1)} is
viewed as a return to equilibrium, or as a nonlocal perturbation of the ferromagnetic ground state
$\omega_{g} \equiv \otimes_{x \in \mathbf{Z}} P_{x}^{\pm,3}$. In the latter case, $\bar{\omega}$ would be viewed as a non-equilibrium
stationary state. A model of the preparation of the system, suitable as the first step in definition 1.2, consists in the following.

Schematically, a $CaF_{2}$ crystal is placed in a magnetic field, thus determining the z-direction. Subsequently, a rf-pulse is
applied to the sample, turning the net nuclear magnetization in the x-direction. The latter is then measured as a function of
time: an oscillatory function slowly damping to zero (Fig. 1.1 in \cite{LoNo} in the original experiment). See also \cite{Em1}.

\end{remark}

\section{Stability of the second law under interactions with the environment}

\subsection{The many histories picture strictly within Schr\"{o}dinger quantum mechanics, K systems and their physical relevance}

In this section we introduce the many-histories picture strictly within Schr\"{o}dinger quantum mechanics, which will be used
in our main theorem in this section. In the context of infinite systems, their introduction is due to Narnhofer and Thirring
\cite{NTh1}, \cite{NTh2}, whose motivation was the proof of the macroscopic purification of states observed in Nature as a consequence 
of interactions with the environment. For this purpose, they needed the concept of (quantum) K system, introduced by Emch \cite{Em2} 
and themselves (see \cite{NTh3} and references given there, as well as \cite{Th}). 

A K system is a standard, and central, concept in the theory of classical dynamical systems (see, e.g., \cite{Wa}, p.101, Def. 4.7).
A prototype of a K system is the baker's map (\cite{LaMa}, p. 74), which is also an Anosov system \cite{Ano}, and displays
exponential instabilities leading to the mixing property \eqref{(1.8)}. Such systems are ``memoryless'', i.e., they posess a
stochastic character which justifies equilibrium statistical mechanics. The microscopic mechanism of this loss of memory is
the sensitive (exponential) dependence on initial conditions, produced by the ``defocalizing shocks'' between the particles
(e.g., gas molecules) occurring in caricatures of interacting classical systems, such as the hard sphere gas, see (\cite{Pen1},p.1950)
for a pedagogic introduction and early references.

An algebraic K system is a family of subalgebras ${\cal A}_{n}, n \in \mathbf{Z}$ of an algebra ${\cal A}$, with the properties:
(i) ${\cal A}_{n+1} \supseteq {\cal A}_{n}$;
(ii) $\bigcup_{n} {\cal A}_{n} = {\cal A}$;
(iii) $\bigcap_{n} {\cal A}_{n} = z \mathbf{1}$
and an automorphism $\sigma$ of ${\cal A}$ with $\sigma({\cal A}_{n}) = {\cal A}_{n+1}$. $\bigcap_{n}$ means the set-theoretic
intersection.

As in \cite{NTh3}, we take ${\cal A}_{n}$ as von Neumann algebras, with $\bigcup_{n}$ meaning algebraic union together with strong 
closure in the representation with a given invariant state: the von Neumann K-system $({\cal A}_{n}, \sigma, \omega)$. The 
isomorphism $\sigma:{\cal A}_{n} \to {\cal A}_{n+1}$ has, by (iii), no nontrivial $\sigma$- invariant subalgebra 
${\cal B} \subset {\cal A}_{n}$. We identify $\sigma^{-1}$ with the time evolution, and assume that $\omega$ is a KMS state with
modular automorphism $\sigma^{-t}$, i.e.,
\begin{equation}
\label{(2.3.1)}
\omega(AB) = \omega(\sigma^{i}(B)A) \forall A,B \in {\cal A}
\end{equation}
We have to refer to \cite{BRo2}, p. 84, for the above terminology, but remark that it will not be used in the sequel.

In the next subsection we shall be concerned with the above mentioned generalization in \cite{NTh1} and \cite{Th} of the many-histories
interpretation of quantum mechanics (\cite{Grif}, \cite{GMH}, \cite{Om}) to (infinite) von Neumann K systems 
$({\cal A}_{n}, \sigma, \omega)$ with a set of projections
$P_{\alpha} \in {\cal A}, \alpha \in \{1, \cdots, r\}$. The operator $\sigma^{-1}$ assigns to each $P_{\alpha}$ a ``time evolution''
$\sigma^{-t}(P_{\alpha}) = P_{\alpha}(t)$ and, in this way, a sequence of ``events'' 
$\{P_{\alpha_{1}}(t_{1}), P_{\alpha_{2}}(t_{2}), \cdots P_{\alpha_{n}}(t_{n})\}$ (a ``history'', briefly written $\underline{\alpha}$
for the index set or the corresponding vector) and a probability distribution over the set $\{\underline{\alpha}\}$ of histories
\begin{equation}
\label{(2.3.2)}
W(\underline{\alpha}) \ge 0 \mbox{ with } \sum_{\underline{\alpha}}W(\underline{\alpha}) = 1
\end{equation}
and such that
\begin{equation}
\label{(2.3.3)}
W(\underline{\alpha}) \equiv \omega(P_{\alpha_{1}}(t_{1})\cdots P_{\alpha_{n}}(t_{n}) \cdots P_{\alpha_{1}}(t_{1}))
\end{equation}
The analogue of \eqref{(2.3.3)} for systems with finite number of degrees of freedom is
\begin{eqnarray*}
W(\underline{\alpha}) = Tr P_{\alpha_{n}}(t_{n}) \cdots P_{\alpha_{1}}(t_{1}) \rho P_{\alpha_{1}}(t_{1}) \cdots P_{\alpha_{n}}(t_{n})\\
= Tr \sqrt{\rho}P_{\alpha_{1}}(t_{1}) \cdots P_{\alpha_{n}}(t_{n})P_{\alpha_{n-1}}(t_{n-1}) \cdots P_{\alpha_{1}}(t_{1}) \sqrt{\rho}  
\end{eqnarray*}
As Wightman observes in his review of the quantum theory of measurement \cite{Wight1}, ``The above formulae were already written
down by Aharonov, Bergman and Lebowitz \cite{ABL} and used by them in a discussion of time reversal invariance in quantum 
measurement theory. These authors regarded the formulae as standard quantum mechanics''.

For the infinite systems the density matrix does not exist, but the notion of a state as a positive linear functional $\omega$ 
as generalization of $Tr \sqrt{\rho} A \sqrt{\rho} = \omega(A)$ carries over, with \eqref{(2.3.3)} replacing the standard
formulae. 

What we need for the main theorem of this section, Theorem 4.1, is the above framework for infinite K systems and histories. 
The Narnhofer-Thirring model used there is a K system (trivially) and the framework applies, as long as we have to do only with finite 
times, which is the case, there, too. The following results, which we present here for completeness, 
have to do with the limit of infinite times. They are, however, crucial, for the purification result of \cite{NTh1}, which is 
the most important physical application of the theory.

The probability distributions \eqref{(2.3.3)} have the properties:
(a.) \emph{decoherence}
\begin{eqnarray*}
\omega(P_{\alpha_{1}}(t_{1})\cdots P_{\alpha_{n}}(t_{n}) \cdots P_{\alpha^{'}_{1}}(t_{1}) \cdots P_{\alpha^{'}_{n}}(t_{n}))\\
\to \delta_{\underline{\alpha},\underline{\alpha^{'}}} \mbox{ as } t_{i}-t_{i+1} \to \infty \forall i
\end{eqnarray*}
(b.) \emph{symmetry}
$W(\alpha_{1}, \cdots, \alpha_{n})$ tends to a symmetric function in the indices $\alpha_{1}, \alpha_{n}$ as
$t_{i}-t_{i+1} \to \infty \forall i$.

Property (a.) is a consequence of the fact that K systems are asymptotically abelian, i.e., classical for large times (see \cite{Th} and
references given there) and property (b.) means that for long times the system forgets all causal links, in analogy to the ``memoryless''
character of classical K systems.

We have the following fundamental theorem of Narnhofer and Thirring (\cite{NTh3}, \cite{Th}):

\emph{ Narnhofer-Thirring Theorem}

Let $({\cal A}_{n},\sigma,\omega)$ be a von Neumann K system, with $\sigma^{-1}$ the time evolution. given a set 
$P_{1}, \cdots, P_{r} \in {\cal A}$, $\epsilon >0$ and $n \in \mathbf{N}$, $n < \infty$, there exists a $T< \infty$ such that for
each history $W(\underline{\alpha})$ given by \eqref{(2.3.3)},
\begin{equation}
\label{(2.3.4)}
|W(\underline{\alpha})- \prod_{i=1}^{n} \omega(P_{\alpha_{i}})| < \epsilon
\end{equation}
whenever $t_{i+1}-t_{i} > T \forall i \in [1,n]$.

It is to be emphasized that Narnhofer and Thirring \cite{NTh1} and Thirring \cite{Th} use the history interpretation as a simple
description of the essential effect of a measurement, but strictly within the Copenhagen interpretation. This follows because only
the so-called \emph{consistent histories} are considered, i.e., such that
\begin{equation}
\label{(2.3.5)}
\omega(P_{\underline{\alpha}^{'}}P_{\underline{\alpha}}) = \delta_{\underline{\alpha}^{'},\underline{\alpha}}
\end{equation}
This is necessary for an interpretation in terms of classical probabilities and is satisfied as soon as \eqref{(2.3.4)} holds.

It may be important, to avoid any confusion, to remark that Griffiths \cite{Grif}, see also (\cite{Grif1}, \cite{Grif2}) bases his
interpretation of quantum mechanics on the \emph{probability} theory of consistent families of histories. This approach has 
resolved some important paradoxes in the conventional theory, see in this connection \cite{Grif2}. We thank Professor Griffiths
for these remarks.   

\subsection{Non-automorphic interactions with the environment}

No physical system is entirely isolated, although isolation may be brought about, in principle, to an arbitrary degree of accuracy.
Even concerning the Universe - the prototype of an isolated system - there will be ``events'', i.e., interactions with the environment,
which may be modelled by the state collapse of a given state, that is, non-automorphic evolutions. One such model is the Narnhofer-Thirring
model (\cite{NTh1}, \cite{NTh2}).

As a physical motivation, consider the question of the macroscopic purification of states by interactions with the environment \cite{NTh1},
such as a magnet below the transition temperature (for a rigorous model of this situation, see Theorem 6.2.48, p.320, of \cite{BRo2}).
It is known that domains of such a magnet are, in Nature, found in a definite direction. To quote \cite{NTh1} (see also \cite{NarWre}): 
``even if nobody looks at them, there will be enough 'events' (i.e., interactions with the environment) to purify the state over the
classical part''. Such are prototypical of the unavoidable interactions of an (even to a good approximation ``closed'') system with
the environment, and were first shown to drive the state into a factor state by Narnhofer \cite{Narn} and Narnhofer and Robinson \cite{NarnRob}.
We refer to Wightman's review article for superselection sectors due to environmental interactions of mesoscopic systems \cite{Wight}.

We consider as the classical quantity of the model the mean magnetization
\begin{equation}
\label{(5.1)}
\vec{m} = \lim_{N \to \infty} \frac{1}{2N} \sum_{x=-N}^{N} \vec{\sigma}_{x}
\end{equation}
Classically, pure states are those in which all spins except a finite number of them point in the same direction $\vec{m}$. Below the
phase transition temperature the state $\omega$ is a mixture of pure, or, more generally, extremal invariant states $\omega_{1}$,
$\omega_{2}$ (see \cite{BRo2}, p.320, Theorem 6.2.48):
\begin{equation}
\label{(5.2)}
\omega = \mu \omega_{1} + (1-\mu) \omega_{2} \mbox{ with } 0<\mu <1
\end{equation}
By \eqref{(2.3.3)}, the histories are also convex combinations
\begin{equation}
\label{(5.3)}
W_{\omega}(\underline{\alpha}) = \mu W_{\omega_{1}}(\underline{\alpha}) + (1-\mu)W_{\omega_{2}}(\underline{\alpha})  
\end{equation}
It was proved in \cite{NTh1} that, for long histories $n \to \infty$, and $t_{i+1}-t_{i} \to \infty$, and a dynamics described by
(von Neumann) K-system they purify in the sense that in \eqref{(5.3)} either $W_{\omega_{1}}$ or $W_{\omega_{2}}$ dominates, i.e.,
$\frac{W_{\omega_{i}}}{W_{\omega_{j}}} < \epsilon$ for arbitrarily small $\epsilon$, and $i \neq j$: which one dominates depends on
the history, although one cannot dominate over the other for all histories since
\begin{equation}
\label{(5.4)}
\sum_{\underline{\alpha}} W_{\omega_{1}}(\underline{\alpha}) = \sum_{\underline{\alpha}} W_{\omega_{2}}(\underline{\alpha}) = 1
\end{equation}
As in \cite{NTh1}, consider the simplest nontrivial case with two states $\omega_{1,2}$ and two projectors $P$, $1-P$, 
and denote $\omega_{1,2}(P)= p_{1,2}$. The model of a measuring apparatus will be defined next, following \cite{NTh1}. The
algebra is the spin algebra ${\cal A}= \{\vec{\sigma}_{x}\}_{x\in \mathbf{Z}}$. The shift $\vec{\sigma}_{x} \to \vec{\sigma}_{x+1}$
is taken as a discrete evolution. We first describe the system $S$. Let $\Omega_{1,2}$ denote two disjoint states
\begin{eqnarray*}
|\vec{n}): \Omega_{1}) = \otimes_{x = -\infty}^{\infty} |\vec{n})_{x} \\
\mbox{ with } \vec{\sigma}_{x} . \vec{n}_{x} |\vec{n}_{x})_{x} = |\vec{n}_{x})_{x} \mbox{ with } \vec{n}_{x}^{2} = 1; \\
|\vec{m}): \Omega_{2}) = \otimes_{x = -\infty}^{\infty} |\vec{m})_{x} \mbox{ with } \vec{n} \neq \vec{m}
\end{eqnarray*}
The two corresponding representations are denoted $\Pi_{1}, \Pi_{2}$ on (incomplete tensor product) Hilbert spaces ${\cal H}_{1}$
and ${\cal H}_{2}$, respectively. the corresponding mean magnetizations are
\begin{equation}
\label{(5.5)}
\vec{M}_{\vec{n}} = \lim_{N \to \infty} \frac{1}{2N+1} \sum_{x=-N}^{N} \Pi_{1}(\vec{\sigma}_{x}) = \vec{n} \mathbf{1}
\end{equation}
and
\begin{equation}
\label{(5.6)}
\vec{M}_{\vec{m}} = \lim_{N \to \infty} \frac{1}{2N+1} \sum_{x=-N}^{N} \Pi_{2}(\vec{\sigma}_{x}) = \vec{m} \mathbf{1}
\end{equation}
The mixed state 
\begin{equation}
\label{(5.8)}
\omega \equiv \mu \omega_{1} + (1-\mu) \omega_{2}
\end{equation}
is obtained by a vector in the orthogonal sum of $\Pi_{1}$ and $\Pi_{2}$:
\begin{eqnarray*}
|\Omega_{S}) = \sqrt(\mu) |\Omega_{1}) \oplus \sqrt(1-\mu) |\Omega_{2}) \in {\cal H}_{1} \oplus {\cal H}_{2} \equiv {\cal H}_{S}\\
\mbox{ with } \Pi = \Pi_{1} \oplus \Pi_{2} \mbox{ and } (\Omega_{S}, \Pi(\vec{\sigma}_{x}) \Omega) = \mu \vec{n} + (1-\mu) \vec{m}
\end{eqnarray*}
The corresponding representation $\Pi$ is reducible: there are two ``superselection sectors'' \cite{Wight}, the mean magnetization
$$
\vec{M} = \lim_{N \to \infty} \frac{1}{2N+1} \sum_{x=-N}^{N}\Pi(\vec{\sigma})_{x}) = \mu \vec{n} \mathbf{1}_{1} + (1-\mu)\vec{m}\mathbf{1}_{2}
$$
lies in the center $Z= \Pi({\cal A}) \cap \Pi({\cal A})^{'}$ , where the prime denotes the commutant, i.e., the set of bounded operators
in the representation space which commute with $\Pi({\cal A})$, which is not a multiple of unity.
For the measuring device (``apparatus'') $A$, Narnhofer and Thirring also use a Hilbert space description, representing a device
measuring one of two spin directions $\vec{s}$ and $\vec{-s}$ (``up'' or ``down''): the state of the device measuring the direction
of $\vec{\sigma}_{i}$ may be represented by a two-dimensional vector
$$
\begin{bmatrix}
u_{i} \\
d_{i}
\end{bmatrix} 
$$
(pointer up and down). The measuring array is again an infinite tensor product $\otimes_{i}\begin{bmatrix}
                                                                                                 u_{i} \\
                                                                                                 d_{i} 
                                                                                           \end{bmatrix}$ 
which belongs to ${\cal H}_{A}$ ( where $A$ stands for ``apparatus''), and we start with a state $|\Omega)_{A}$ where all $u_{i}$ are
zero. For the time evolution, we take a shift $U$, and then consider an instantaneous measurement of the direction $\vec{s}$ of
a spin, the corresponding projector being
$$
P_{k} = \frac{1}{2} (1+\vec{\sigma}_{k} \cdot \vec{s}) = |s)_{k} \times (|s)_{k})^{\dag}
$$
If the answer is one, we have the pointer unchanged, if the answer is zero, we turn the pointer up. The turning of the pointer is 
effected by an operator $\tau$,
$$
\tau_{k} \begin{bmatrix}
               u_{i} \\
               d_{i} 
         \end{bmatrix}= \begin{bmatrix} 
                               d_{i} \\
                               u_{i}
                        \end{bmatrix}
$$
Thus, the effect of measuring $\vec{\sigma}_{1}$ is
$$
V_{1} = P_{1}+(1-P_{1})\tau_{1}
$$
or, written in full detail, with operators on ${\cal H} = {\cal H}_{S} \otimes {\cal H}_{A}$,
\begin{equation}
\label{(5.9)}
V_{1} = (\Pi_{1}(P_{1}) \oplus \Pi_{2}(P_{1})) \otimes \mathbf{1}+(1-\Pi_{1}(P_{1})) \oplus ((1-\Pi_{2}(P_{1})))\otimes \tau_{1}
\end{equation}
The time evolution $U$ between the measurements shifts by one unit
\begin{equation}
\label{(5.10)}
U \Pi_{1,2}(P_{k}) = \Pi_{1,2}(P_{k+1})U \mbox{ together with } U\tau_{k}= \tau_{k+1}U
\end{equation}
so that the full time evolution of $|\Omega) = |\Omega_{S} \otimes |\Omega)_{A}$ after $n$ time units is
\begin{equation}
\label{(5.11)}
|\Omega(n)) = V_{1}UV_{1}U \cdots V_{1}U|\Omega)=V_{1}V_{2} \cdots V_{n}|\Omega)
\end{equation}
since $U^{k}|\Omega)=|\Omega)$. The results of the measurements are encoded in the ${\cal H}_{A}$-part of $|\Omega(n))$, so we
decompose $|\Omega(n))$, given by \eqref{(5.11)} in an orthogonal basis of ${\cal H}_{A}$,
\begin{equation}
\label{(5.12)}
|\Omega(n)) = \sum_{\alpha_{i}=0}^{1} v(\underline{\alpha}) \tau_{1}^{\alpha_{1}} \cdots \tau_{n}^{\alpha_{n}}|\Omega_{A})
\end{equation}
with $v(\underline{\alpha}) \in {\cal H}_{S}$. Wherever $\alpha_{i}=0$, the corresponding spin is in direction $\vec{s}$,
for $\alpha_{i}=1$ we have $\vec{-s}$:
\begin{equation}
\label{(5.13.1)}
v(\underline{\alpha}) = v_{1}(\underline{\alpha}) \oplus v_{2}(\underline{\alpha})
\end{equation}
together with
\begin{equation}
\label{(5.13.2)}
P_{1}|n)_{1} = |s)_{1}(s|n)
\end{equation}
and
\begin{equation}
\label{(5.13.3)}
(1-P_{1})|n) = |-s)(-s|n)
\end{equation}
If we introduce $|(s|n)|^{2}=p_{1}$, thus $|(-s|n)|^{2} = 1-p_{1}$, and similarly  $|(s|m)|^{2}=p_{2}$, $|(-s|m)|^{2} = 1-p_{2}$,
and if $\alpha$ contains $l$ zeros and $(n-l)$ ones, we have
\begin{equation}
\label{(5.13.4)}
||v_{1}(\underline{\alpha})||^{2} = \mu p_{1}^{l}(1-p_{1})^{n-l} ; ||v_{2}(\underline{\alpha})||^{2} =(1-\mu) p_{2}^{l}(1-p_{2})^{n-l}
\end{equation}
together with
\begin{equation}
\label{(5.14)}
W(\underline{\alpha}) = ||v(\underline{\alpha})||^{2}
\end{equation}
Note that there is no collapse of the wave-function after each measurement - the evolution is unitary and only at the end 
the configuration of classical pointers is read. We shall refer to the above model as the \emph{Narnhofer-Thirring model}.

\subsection{Invariance of the mean entropy on the average and comparison with finite systems}

We are now ready to present our main results. Firstly, let us examine finite systems or, more precisely, systems with 
a finite number of degrees of freedom. We  consider, as in \cite{NarWre}, a finite composite system consisting 
of a finite system together with a measurement apparatus, for example a finite version of the Narnhofer-Thirring model, 
described by a collection of projectors
$\{P_{\alpha}\}$, $\alpha \in [1,n]$, such that
\begin{equation}
\label{(5.15)}
P_{\alpha}P_{\alpha^{'}} = P_{\alpha} \delta_{\alpha,\alpha^{'}}
\end{equation}
Inspired by von Neumann (Chap. V. of \cite{vN}, see also the discussion in Wightman (\cite{Wight1}, p. 431) concerning degenerate
spectra), it is natural to take the $\{P_{\alpha}\}$ ($\alpha=1, \cdots, n$) as ``macroscopic observables'', and define a density
matrix $\rho$ as \emph{decoherent} with respect to the set $\{P_{\alpha}\}$ if
\begin{equation}
\label{(5.16)}
P_{\alpha^{'}} \rho P_{\alpha} =0 \mbox{ if } \alpha \ne \alpha^{'}
\end{equation}
When a measurement is performed and the position of the pointer was observed, the probability of its position can only be evaluated
if we repeat the measurement. In the individual observation it has a definite value, and the measured system is in a pointer position:
reduction of the wave-function turns into the collapse of the wave-function. We have gained knowledge, not necessarily in the 
individual case, but on the average, and, indeed, by Lemma 3 of \cite{NarWre}, the average von Neumann entropy $S_{av}$, defined by
\begin{equation}
\label{(5.17)}
S_{av} = \sum_{\alpha=1}^{n} Tr(\rho P_{\alpha}) S(\rho_{\alpha})
\end{equation}
where
\begin{equation}
\label{(5.18)}
S(\rho) = -Tr (\rho \log(\rho))
\end{equation}
(with $k_{B}=1$) is the von Neumann entropy, is \emph{reduced}; above,
\begin{equation}
\label{(5.19)}
\rho_{\alpha} \equiv [Tr (\rho P_{\alpha})]^{-1} P_{\alpha} \rho P_{\alpha}
\end{equation}
is one of the collapsed states, obtained with probability $Tr(\rho P_{\alpha})$ through unitary evolution of an initial decoherent state,
i.e., which satisfies \eqref{(5.16)}. This follows from the strict concavity of the entropy of a finite system, i.e., 
\begin{equation}
\label{(5.20)}
S_{\Lambda}(\alpha \rho_{\Lambda}^{1}+(1-\alpha) \rho_{\Lambda}^{2}) > \alpha S_{\Lambda}(\rho_{\Lambda}^{1}) + (1-\alpha) S_{\Lambda}(\rho_{\Lambda}^{2})
\end{equation}
 The mean entropy is, however, \emph{affine}(\cite{LanRo}):
\begin{equation}
\label{(5.21)}
s(\alpha \omega_{1} + (1-\alpha) \omega_{2}) = \alpha s(\omega_{1}) + (1-\alpha) s(\omega_{2})
\end{equation}

For the statistical-thermodynamical significance of the property of affinity, see \cite{Sewell1}, p.60.

We now consider the Narnhofer-Thirring model. After $n$ (time) steps, we use the notation
\begin{equation}
\label{(5.22)}
\underline{\alpha}_{n} = (\alpha_{1}, \cdots, \alpha_{n})
\end{equation}

\begin{theorem}
\label{th:5.1}
Let $\omega$ be the state \eqref{(5.8)}. Then, in the Narnhofer-Thirring model with initial state $\omega$, the mean entropy
is conserved on the average by the collapse of $\omega$:
\begin{equation}
\label{(5.23)}
s_{av} = s(\omega)
\end{equation}
where
\begin{equation}
\label{(5.24)}
s_{av} \equiv \sum_{\underline{\alpha}_{n}} W_{\omega}(\underline{\alpha}_{n}) s(\omega_{\underline{\alpha}_{n}})
\end{equation}
and
\begin{equation}
\label{(5.25)}
\omega_{\underline{\alpha}_{n}} \equiv \omega(P_{\underline{\alpha}_{n}} .)[W_{\omega}(\underline{\alpha}_{n})]^{-1}
\end{equation}

\begin{proof}

Reading the pointer in a position $\underline{\alpha}_{n}$, i.e., after $n$ time steps, the initial state collapses to
\eqref{(5.25)}. The property of affinity \eqref{(5.21)} now implies that
\begin{eqnarray*}
s_{av} = \sum_{\underline{\alpha}_{n}} W_{\omega}(\underline{\alpha}_{n}) s(\omega_{\underline{\alpha}_{n}}) = \\
= s(\sum_{\underline{\alpha}_{n}} W_{\omega}(\underline{\alpha}_{n}) [W_{\omega}(\underline{\alpha}_{n})]^{-1} \omega(P_{\underline{\alpha_{n}}})\\
= s(\omega)
\end{eqnarray*}
by linearity of $\omega$ and the fact that $\sum_{\underline{\alpha}_{n}} P_{\underline{\alpha_{n}}} = 1$.

\end{proof}
\end{theorem}

As a consequence of Theorem ~\ref{th:5.1}, we have:

\begin{corollary}
\label{cor:5.1}
The statistical thermodynamical version of the second law of thermodynamics Theorem ~\ref{th:4.1} is \emph{stable} under non-automorphic
interactions with the environment, exemplified by the Narnhofer-Thirring model, in sharp contrast to the second law of thermodynamics
for the quantum Boltzmann entropy \cite{NarWre}.
\end{corollary}

The last sentence of Corollary ~\ref{cor:5.1} follows from lemma 3 of \cite{NarWre}.

\subsection{Physical interpretation}

In his paper ``Against measurement'', John Bell \cite{Bell} insisted on the necessity of physical precision regarding such words as
reversible, irreversible, information (whose information? information about what?).

In the adiabatic transformation there is a first step, a finite preparation time during which external forces act, at the end
of which the Hamiltonian associated to the initial equilibrium state is restored, and remains so ``forever'' during the
second step. The dynamics of preparation of the state is therefore *not* time-reversal invariant, leading to the \emph{time arrow}
mentioned in remark 1.1. 

Given a time arrow, the process $\omega_{1}(0) \to \omega_{2}(\infty)$ is reversible (irreversible) iff the inverse process 
$\omega_{2}(0) \to \omega_{1}(\infty)$ is possible (impossible). The first alternative takes place iff $s(\omega_{1}) = s(\omega_{2})$,
the second one iff $s(\omega_{1}) < s(\omega_{2})$. Infinite time $t=\infty$ means that $t$ is much larger than a quantity $t_{D}$,
the relaxation time, or decoherence time in measurement theory. Of course, irreversibility is incompatible with time-reversal invariance,
because the entropy or the mean entropy cannot both strictly increase and strictly decrease with time. This is the Schr\"{o}dinger paradox
\cite{Schr}, cited in Lebowitz's inspiring review of the issue of time-assymetry \cite{Leb}.

Entropy $S_{\Lambda} = |\Lambda|\log D -I_{\Lambda}$, with $I_{\Lambda}$ denoting the (quantum) information.
For quantum spin systems $0 \le S_{\Lambda}/|\Lambda| \le \log D$, and therefore $0 \le I_{\Lambda}/|\Lambda| \le \log D$. It attains
its maximum value for pure states, which are characterized by $S_{\Lambda} = 0$. Under ``collapse'', each collapsed state is pure
and therefore information is gained: this explains that the (Boltzmann and von Neumann) entropies are reduced, on the average,
violating the second law (on the average).
Equivalently, entropy is, in Boltzmann's sense, a measure of a macrostate's wealth of ``microstates'', and therefore grows by mixing, 
but it turns out that this growth is *not* extensive and disappears upon division by $|\Lambda|$ , i.e., taking the infinite volume limit 
(inequalities of Lanford and Robinson \cite{LanRo}), so that the affinity property \eqref{(5.21)} results and, with it, the stability of 
the second law in the form of Theorem 3.1 under interactions with the environment (Corollary ~\ref{cor:5.1}).

\section{Conclusion}

One central and dominating feature of the analysis over finite vs infinite dimensional spaces is that in the infinite dimensional 
case the solution may depend \emph{discontinuously} on the parameters of the problem. Indeed, infinite systems may exhibit 
\emph{singularities}, not present in finite macroscopic systems, well-known in the theory of \emph{phase transitions}: they
are parametrized by \emph{critical exponents}, which, moreover, display \emph{universal} properties, in excellent agreement with experiment!
Another example is furnished by the upper-semicontinuity of the mean entropy (versus continuity of the finite entropy).

In greater generality, the physical ``N large but finite'' differs qualitatively from ``N infinite'' because the latter exhibits 
\emph{universal} properties not found in finite systems. Two examples of these universal properties, crucial in our approach, are the 
upper semicontinuity and affinity of the mean entropy, whose finite-volume counterparts (continuity of $\frac{S_{\Lambda}}{|\Lambda|}$
and strict concavity of $\frac{S_{\Lambda}}{|\Lambda|}$) are not universal because, not being uniform in $|\Lambda|$, they depend on 
the volume $|\Lambda|$ of the system. The fact that (only) ``N infinite'' is in good agreement 
with experiment is explained by the fact that, with $N \approx 10^{24}$, macroscopic systems are extremely close to infinite systems 
(the success of the thermodynamic limit!)

It is well known that the thermodynamic entropy, which is relevant to the second law, does not, in general, coincide with the
(microscopic) Gibbs-von Neumann entropy, except for a limited class of states which includes the standard equilibrium states.
Our approach introduces two types of coarse-graining in the entropy; the first, in space, by building the mean entropy.
This enables consideration of states of infinite systems which may (and do sometimes) have nontrivial ergodic properties: starting from 
an unstable state at an initial time $t=0$, they may approach a different (``equilibrium'') state as $t \to \infty$ - i.e., there is
also a ``coarse graining'' in time. The important issue here is the fact that a sequence of pure states may tend to a mixed state,
a fact first observed by Hepp \cite{Hepp}, see also \cite{NTh1}, and valid only for states of infinite systems. In the examples given, 
the first step of the adiabatic transformation - which is the thermodynamically
nontrivial step - yields no change in the mean entropy (it remains equal to the initial value of zero), it is only the second step, which
is thermodynamically trivial - a relaxation - which yields a change in the specific entropy. This feature is not realistic and is due
to the extreme scarcity of models in which the time evolution of states is under control. Both theorems
3.1 and 4.1 are, however, of general validity, because of the universal properties mentioned above, and of the natural type of coarse-graining
associated to the use of the mean entropy (a density).

There is what seems to us to be a close analogy with the situation in classical dynamical systems as regards the intuitive reason
for this possible growth of the mean entropy with time. If one looks at the entropy $S(U_{t},V_{t})$ of a system with finite number
of degrees of freedom, $(U_{t},V_{t})$ are analogous to the individual ``trajectories'' of a classical map, which, for a (time-reversible)
K system, such as the baker map (\cite{LaMa}, Fig. 1.2.4 and section 1.3) are highly erratic and do not tend to any limit as $t \to \infty$.
On the other hand, in $s(u_{t},v_{t})$, where $s$ is the mean entropy and $u,v$ the energy and volume densities, the $u,v$ densities
may be expected to behave qualitatively as the density $\rho_{t}$, e.g., in the baker transformation (\cite{LaMa}, Ex. 4.1.3, p. 48),
\begin{equation}
\label{(6.1)}
\lim_{t \to \infty} \rho_{t} = \bar{\rho} \equiv 1
\end{equation}
\eqref{(6.1)} is analogous to (129), where, for $t=n=1,2, \cdots$, i.e., discrete time units,
\begin{equation}
\label{(6.2)}
\rho_{n} = P^{n} \rho_{0}
\end{equation}
and $P$ denotes the Ruelle-Perron-Frobenius operator \cite{LaMa} which maps densities to densities, and $\rho_{0}$ denotes an arbitrary
initial density
\begin{equation}
\label{(6.3)}
\rho_{0} \in L^{1}(X), \rho_{0} \ge 0, \int_{X} \rho_{0}(x) dx = 1
\end{equation}
with $X$ denoting the unit square in $\mathbf{R}^{2}$. By \eqref{(6.3)}, $\rho_{0}$ cannot be a delta function, i.e., individual trajectories
are excluded. $P$ is a Markov operator \cite{LaMa}, with a unique fixed point, which is the function
\begin{equation}
\label{(6.4)}
 \bar{\rho} = 1 \forall x \in X
\end{equation}
(analogous to the uniform ``microcanonical'' distribution on the square). The reduction of information which results from building a
density provides a ``restoring force'' towards equilibrium \eqref{(6.4)}, with maximal disorder or entropy. See also \cite{Wresz}.

Our main new result in this paper is that non-automorphic interactions with the environment, although known to produce 
on the average a strict reduction of the Boltzmann entropy of systems with finite number of degrees of freedom (lemma 3 of \cite{NarWre}), 
are proved to conserve the mean entropy on the average, as a consequence of the latter's property of affinity. As a consequence, we
do not need to assume that such interactions are rare on the thermodynamic time scale, in order to assure the validity of the second
law (see the remarks on the last paragraph of \cite{NarWre})and \emph{stability} with regard to interactions with the environment 
follows instead (Corollary ~\ref{cor:5.1}). We view this as an additional important indication of the naturalness of the present
approach, suggesting that it may turn out to be of general validity, including classical mechanics (see the introduction) and
relativistic quantum field theory, but serious open problems of technical nature remain in both cases.    

\emph{Acknowledgements}

I am greatly indebted to Bruno Nachtergaele, who provided me with two crucial elements of the proof of Theorem 4.1: his theorem (together
with Sims and Young) on local approximations to automorphisms of local observables, and Fannes' beautiful
inequality.

I also thank very heartily Heide Narnhofer for the very enjoyable and instructive collaboration \cite{NarWre}.

The present paper owes very much to several people: Elliott Lieb, for his encouraging remarks on the related \cite{Wresz}, Derek Robinson
for a correspondence in which he mentioned the property of affinity, Prof. D. Ruelle for his encouraging remarks in correspondence and for
asking the ``right questions'' in \cite{Ru}, Professor Oliver Penrose for correspondence on the Penrose-Gibbs theorem, and Geoffrey
Sewell for helping me to formulate the time-arrow problem in a most concise way.

I should also like to thank an anonimous referee for helpful remarks concerning a previous version of Definition 1.2.

\qquad
\textbf{Statement concerning data availability}
\textbf{The author confirms that all the data supporting the findings of this study are available within the article.}
\qquad


\begin{thebibliography}{99}

\bibitem[Wre120]{Wre1}
W.~F. Wreszinski.
\newblock Irreversibility, the time arrow and a dynamical proof of the second law of thermodynamics.
\newblock {\em Quantum Stud: Math. and Found.}, 7:125-136, 2020.

\bibitem[NW14]{NarWre}
H.~Narnhofer and W.~F. Wreszinski.
\newblock On reduction of the wave packet, decoherence, irreversibility and the second law of thermodynamics.
\newblock {\em Phys. Rep.}, 541:249-278, 2014.

\bibitem[Wehrl78]{Wehrl}
A.~Wehrl.
\newblock General properties of entropy.
\newblock {\em Rev. Mod. Phys.}, 50:221, 1978.

\bibitem[Sew86]{Sewell1}
G.~L. Sewell.
\newblock Quantum theory of collective phenomena.
\newblock Oxford University Press, 1986.

\bibitem[LY99]{LY}
E.~H. Lieb and J.~Yngvason.
\newblock The physics and mathematics of the second law of thermodynamics.
\newblock {\em Phys. Rep.}, 310:1-96, 1999.

\bibitem[Pen79]{Pen1}
O.~Penrose.
\newblock Foundations of statistical mechanics.
\newblock {\em Rep. Progr. Phys.}, 42:1937, 1979.

\bibitem[Gib60]{Gibbs}
J.~W. Gibbs.
\newblock Elementary principles in statistical mechanics.
\newblock Dover, N. Y. 1960.

\bibitem[Wa75]{Wa}
P.~Walters.
\newblock Ergodic theory - introductory lectures.
\newblock Lect. Notes in Math. 458, Springer 1975.

\bibitem[RR67]{RR}
D.~W. Robinson and D.~Ruelle.
\newblock Mean entropy of states in classical statistical mechanics.
\newblock {\em Comm. Math. Phys.}, 5:288, 1967.

\bibitem[LanRo68]{LanRo}
O.~E. Lanford and D.~W. Robinson.
\newblock Mean entropy of states in quantum statistical mechanics.
\newblock {\em J. Math. Phys.}, 9:1120, 1968.

\bibitem[He72]{Hepp}
K.~Hepp.
\newblock Quantum theory of measurement and macroscopic observables.
\newblock {\em Helv. Phys. Acta}, 45:237, 1972.

\bibitem[Wre20]{Wresz}
W.~F. Wreszinski.
\newblock One or two small points in thermodynamics.
\newblock arXiv 2003.03800, 27-5-2020 (v5).

\bibitem[Ru67]{Ru}
D.~Ruelle.
\newblock States of classical statistical mechanics.
\newblock {\em J. Math. Phys.}, 8:1657, 1967.

\bibitem[PW78]{PW}
W.~Pusz and S.~L. Woronowicz.
\newblock Passive states and KMS states for general quantum systems.
\newblock {\em Comm. Math. Phys.}, 58:273, 1978.

\bibitem[BRo297]{BRo2}
O.~Bratelli and D.~W. Robinson.
\newblock Operator algebras and quantum statistical mechanics, vol.2, 2nd. ed.
\newblock Springer Verlag 1997.

\bibitem[Pe79]{Pe}
R.~E. Peierls.
\newblock Surprises in theoretical physics.
\newblock Princeton University Press 1979.

\bibitem[BB92]{BB}
P.~Blanchard and E.~Br\"{u}ning.
\newblock Variational methods in mathematical physics: a unified approach.
\newblock Springer Verlag, Berlin 1992.

\bibitem[Em66]{Em1}
G.~G. Emch.
\newblock Nonmarkovian model of the approach to equilibrium.
\newblock {\em J. Math. Phys.}, 7:1198, 1966.

\bibitem[Ra70]{Ra}
C.~Radin.
\newblock Approach to equilibrium in a simple model.
\newblock {\em J. Math. Phys.}, 11:2945, 1970.

\bibitem[Dy69]{Dy}
F.~J. Dyson.
\newblock Existence of a phase transition in a one-dimensional Ising ferromagnet.
\newblock {\em Comm. Math. Phys.}, 12:91, 1969.

\bibitem[LR72]{LR}
E.~H. Lieb and D.~W. Robinson.
\newblock The finite group velocity of quantum spin systems.
\newblock {\em Comm. Math. Phys.}, 28:251, 1972.

\bibitem[NSY19]{NSY}
B.~Nachtergaele, R.~Sims and A.~Young.
\newblock Quasi-locality bounds for quantum lattice systems I- Lieb-Robinson bounds, quasi-local maps and spectral flow automorphisms.
\newblock {\em J. Math. Phys.}, 60:061101, 2019.

\bibitem[Fa73]{Fa}
M.~Fannes.
\newblock A continuity property of the entropy density for quantum lattice systems.
\newblock {\em Comm. Math. Phys.}, 31:291, 1973. 

\bibitem[Lo57]{LoNo}
I.~J.Lowe and R.~E. Norberg.
\newblock Free induction decays in solids.
\newblock {\em Phys. Rev.}, 107:46, 1957.

\bibitem[Cho05]{Cho}
H.~Cho, T.~D. Ladd, J.~Baugh, D.~G. Cory, and C.~Ramanathan.
\newblock Multi-spin dynamics of the solid state NMR free induction decay.
\newblock {\em Phys. Rev. B}, 72:054427, 2005.

\bibitem[Ma12]{Maze}
J.~R. Maze, A.~Dr\'{e}au, V.~Waselowski, H.~Duarte, J.~F. Roch and V.~Jacques.
\newblock Free induction decays of single spins in diamond.
\newblock {\em New Jour. Phys.}, 14:103041, 2012

\bibitem[NTh99]{NTh1}
H.~Narnhofer and W.~Thirring.
\newblock Macroscopic purification of states by interactions.
\newblock in A.~Amann, H.~Armanspacher and U.~M\"{u}ller-Herrold (eds.),On quanta, mind and matter, Kluwer Academic publishers 1999.

\bibitem[NTh96]{NTh2}
H.~Narnhofer and W.~Thirring.
\newblock Why Schr\"{o}dinger's cat is most likely to be alive or dead.
\newblock ESI preprint (1996) 34.

\bibitem[Em75]{Em2}
G.~G. Emch.
\newblock Generalized K flows.
\newblock {\em Comm. Math. Phys.}, 49:191, 1975.

\bibitem[NTh94]{NTh3}
H.~Narnhofer and W.~Thirring.
\newblock Clustering for algebraic K systems.
\newblock {\em Lett. Math. Phys.}, 30:307, 1994.

\bibitem[Th96]{Th}
W.~Thirring.
\newblock The histories of chaotic quantum systems.
\newblock {\em Helv. Phys. Acta}, 69:706, 1996.

\bibitem[Wa75]{Wa}
P.~Walters.
\newblock Ergodic theory - introductory lectures.
\newblock Lect. Notes in Math. 458, Springer 1975.

\bibitem[LM85]{LaMa}
A.~Lasota and M.~C. Mackey.
\newblock Probabilistic properties of deterministic systems.
\newblock Cambridge university press 1985.

\bibitem[Ano67]{Ano}
D.~V. Anosov.
\newblock Ergodic properties of geodesic flows on closed Riemannian manifolds of negative curvature.
\newblock {\em Proc. Steklov Inst. Math.}, 90:1, 1967.

\bibitem[Gr84]{Grif}
R.~B. Griffiths.
\newblock Consistent histories and the interpretation of quantum mechanics.
\newblock {\em J. Stat. Phys.}, 36:219, 1984.

\bibitem[GMH91]{GMH}
M.~Gell-Mann and J.~B. Hartle.
\newblock In: K.~K. Phua and Y.~Yamaguchi, eds., Proc. of the 25th Int. Conf. of High Energy Physics
\newblock World Scientific, Singapore 1991.

\bibitem[Om94]{Om}
R.~Omn\`{e}s.
\newblock The interpretation of quantum mechanics.
\newblock Princeton University Press 1994.

\bibitem[Wi92]{Wight1}
A.~S. Wightman.
\newblock Some comments in the quantum theory of measurement.
\newblock In: Probabilistic methods in mathematical physics, ed. by F.~ Guerra, M.~Ioffredo and C.~Marchioro.
\newblock World Scientific 1992.

\bibitem[ABL64]{ABL}
Y.~Aharonov, P.~G. Bergmann and J.~L. Lebowitz.
\newblock Time symmetry in the quantum process of measurement.
\newblock {\em Phys. Rev.}, 134B:1410, 1964.

\bibitem[Grif19]{Grif1}
R.~B. Griffiths.
\newblock The consistent histories approach to quantum mechanics.
\newblock Stanford Encyclopaedia of Philosophy, 2019.

\bibitem[Grif21]{Grif2}
R.~B. Griffiths.
\newblock Bell inequalities versus quantum physics: a reply to Lambare.
\newblock arXiv 2106.09824, to be published in Phys. Rev. A.

\bibitem[NR75]{NarnRob}
H.~Narnhofer and D.~W. Robinson.
\newblock Dynamical stability and pure thermodynamical phases.
\newblock {\em Comm. Math. Phys.}, 41:89, 1975.

\bibitem[Narn93]{Narn}
H.~Narnhofer.
\newblock Stability of pure thermodynamical phases in quantum statistics.
\newblock In: Phase transitions, ed. by R.~Kotecky, pp. 150-158.
\newblock World Scientific, Singapore, 1993.

\bibitem[Wi95]{Wight}
A.~S. Wightman.
\newblock Superselection sectors: old and new.
\newblock {\em Nuovo Cim. B}, 110:751, 1995.

\bibitem[vN55]{vN}
J.~von Neumann.
\newblock Mathematical Foundations of Quantum Mechanics.
\newblock Princeton University Press, Princeton, NJ, 1996.

\bibitem[Be75]{Bell}
J.~S. Bell.
\newblock Against measurement.
\newblock {\em Phys. World}, Aug.1990.

\bibitem[Schr]{Schr}
E.~Schr\"{o}dinger.
\newblock The Spirit of Science, in What is Life? and Other Scientific Essays.
\newblock Doubleday Anchor Books, Garden City, N. Y., 1965, 229-250.

\bibitem[Leb08]{Leb}
J.~L. Lebowitz.
\newblock Time-asymmetric Macroscopic Behavior: an Overview.
\newblock In Boltzmann's Legacy, ed. by G. Gallavotti et al, Eur. Math. Soc. 2008.

\end{thebibliography}
\end{document}